\documentclass[letterpaper, 10 pt, conference]{ieeeconf}  %

\IEEEoverridecommandlockouts                              %
\usepackage{amsmath} %
\usepackage{amssymb}  %
\usepackage{xcolor}
\usepackage{hyperref}
\usepackage{graphicx}
\usepackage[english]{babel}
\usepackage{tikz}

\definecolor{myblue}{HTML}{0072BD}
\definecolor{myred}{HTML}{D95319}
\definecolor{myyellow}{HTML}{EDB120}
\definecolor{mypurple}{HTML}{7E2F8E}
\definecolor{mygreen}{HTML}{77AC30}
\usepackage{multirow}
\usepackage{aligned-overset}
\usepackage{siunitx}

\usepackage{tikz}

\usetikzlibrary{positioning,shapes.geometric,}
\usepackage[european,americanvoltages,americancurrents]{circuitikz}

\title{\LARGE \bf
Data-driven Linear Quadratic Integral Control: \\
A Convex Formulation and Policy Gradient Approach
}

\author{Armin Gießler, Pol Jané-Soneira, and Sören Hohmann%
\thanks{A. Gießler and S. Hohmann are with the Institute of Control Systems, Karlsruhe Institute of Technology (KIT), 76131 Karlsruhe, Germany. P. Jané-Soneira is with ABB Corporate Research Center, Kallstadter Str. 1, 68309 Mannheim, Germany. Corresponding author is Armin Gießler,
{\tt \footnotesize armin.giessler@kit.edu}.}
}

\newtheorem{theorem}{Theorem}
\newtheorem{proposition}{Proposition}

\newtheorem{remark}{Remark}
\newtheorem{assumption}{Assumption}
\newtheorem{lemma}{Lemma}

\newtheorem{objective}{Objective}

\usepackage[style=ieee, %
backend=bibtex,  %
isbn=false,  %
doi=false, %
giveninits=true,url=false,date=year, maxbibnames=99, minnames=1, maxnames=3, sorting=none,alldates=year,]{biblatex}
\AtEveryBibitem{\clearfield{pages}}\AtEveryBibitem{\clearfield{volume}}\AtEveryBibitem{\clearfield{number}}

\begin{document}

\maketitle
\thispagestyle{empty}
\pagestyle{empty}

\begin{abstract}
This paper studies the data-driven synthesis of linear quadratic integral (LQI) controllers for continuous-time systems. The objective is to achieve optimal state-feedback control with integral action for reference tracking using only measured data. To this end, we derive a data-driven closed-loop parameterization of the augmented dynamics that incorporates the integral state while relying solely on input-state-output measurements of the underlying system. Based on this parameterization, a data-driven convex optimization problem is formulated whose solution yields the optimal linear quadratic regulator (LQR) feedback gain for the augmented system without explicit knowledge of the system matrices. In addition, a projected policy gradient flow, which is a preconditioned version of the model-based policy gradient flow, is derived to compute the optimal controller within the space of stabilizing gains. The proposed approach enables data-driven optimal tracking control while avoiding explicit state augmentation in the data collection phase. For time-varying gains, as they arise in an online implementation, the direct, incremental, and feedforward realizations of the LQI controller are compared, and the equilibrium required for feedforward is obtained from the same data. The method is illustrated on a distributed generation unit (DGU) in a DC microgrid, including an online implementation with closed-loop data under load changes, and by a Monte Carlo study on random systems with noisy data.
\end{abstract}

\section{Introduction}
Recent advances in data-driven control enable the synthesis of  optimal controllers directly from measurement data, bypassing the need for exact system models \cite{de2019formulas}.
Many data-driven approaches consider the  infinite-horizon linear quadratic regulator (LQR) that regulates the state to zero by minimizing an objective that achieves a compromise between transient performance and control effort.
To achieve robust output tracking of a reference signal, an integral action can be incorporated into the LQR control law to complement the proportional action, referred to as linear quadratic integral (LQI) control.
This paper focuses on data-driven LQI control using convex optimization and policy gradient methods, aiming to achieve optimal tracking performance while relying solely on input-state-output measurements.

\subsubsection*{Literature review}
This section reviews relevant literature on LQI and PID synthesis via LQR, linear quadratic tracking, and data-driven convex optimization and policy gradient methods for LQR.

LQI control was first introduced in \cite{YOUNG01051972} to achieve asymptotic tracking of constant reference signals. Since then, it has been widely applied across various domains, including aerospace \cite{purnawan2017design}, power grids \cite{kedjar_lqr_2011}, and process control \cite{MALKAPURE201455} for tracking time-varying reference signals.
At its core, the LQI control law consists of a static state-feedback controller augmented with an integral action that accumulates the tracking error between the output and the reference signal \cite{YOUNG01051972}.
The synthesis of a stabilizing LQI control law can be achieved by augmenting the system dynamics with an integral state and subsequently applying  LQR design methods, such as convex optimization \cite{Feron,Balakrishnan1995} or policy gradient approaches \cite{fazel2018global,Bin2023}.

The LQI controller is often compared to the classical proportional-integral-derivative (PID) controller, as both involve feedback of the integrated tracking error, and a linear system in closed loop with a PID controller can be represented as an augmented linear system under static state feedback \cite{mukhopadhyay1978pid} or static output feedback \cite[Chapter~7]{wang_pid_nodate}.
The structural constraints inherent to PID control, however, make its synthesis challenging \cite{datta1999structure} (see Section~\ref{subsec:PID}), whereas the design of a stabilizing LQI controller reduces to the LQR problem of the augmented system, at the cost of requiring full-state measurements.
Policy gradient solutions for dynamic output feedback without feedthrough that minimize the LQR objective are discussed in \cite{Duan2024}.
This class of output-feedback controllers includes observer-based state-feedback controllers as a special case, but excludes PID controllers that incorporate a feedthrough term.

In contrast to LQI control, linear quadratic tracking (LQT) penalizes the tracking error directly, but its optimal control law contains a feedforward term that is obtained by integrating an adjoint system backward in time, so that the reference trajectory has to be known a priori \cite[Sec.~4.1]{lewis_optimal_2012}. For continuous-time systems with unknown dynamics, reinforcement learning solves the LQT problem for references generated by a command generator \cite{qin_online_2014}, where a discounted \cite{modares_linear_2014} or an average cost \cite{adibyaghmaie_output_2019} admits non-decaying references. For discrete-time systems, policy optimization based on the sample covariance parameterization addresses set-point tracking \cite{kang2025linear} and LQT with a finite preview of the reference \cite{kang2026directdatadrivenlinearquadratic}. LQI control, in turn, is causal, requires no knowledge of future references, and eliminates the effect of constant disturbances by the integral action \cite[Sec.~3.7.2]{kwakernaak_linear_1972}, at the expense of optimality with respect to a tracking cost.

In the context of data-driven optimal control, \cite{de2019formulas} presents a parameterization of
a linear system in closed loop with a state-feedback controller using only input-state data and a matrix satisfying a matrix equality.
An advancement of this parameterization is the sample covariance parameterization presented in \cite{Zhao2025,ZhaoCovariance}, resulting in data matrices independent of the sample size. These representations laid the foundation for numerous data-driven LQR methods, including convex optimization approaches \cite{de2019formulas,lopez2025databasedcontrolcontinuoustimelinear} and policy gradient methods \cite{Zhao2023,Zhao2025}. For noise-free data, these direct designs yield the same gain as the certainty-equivalence design, i.e., the model-based design with the least-squares estimates of the system matrices \cite{dorfler2023}, and for the sample covariance parameterization, this equivalence holds also for data with process noise \cite[Lemma~1]{Zhao2025}. Regularizations that account for noisy data are studied in \cite{dorfler2022,ZhaoCovariance}.

\subsubsection*{Contributions}
The main contributions of this paper are summarized as follows:
\begin{enumerate}
   \item A data-driven closed-loop parameterization for the synthesis of LQI controllers is introduced, enabling the direct design of optimal tracking controllers from measured data of the underlying system without measuring the integral state. The equilibrium required for feedforward is obtained from the same data.
   \item Based on this parameterization, a convex optimization problem and a projected policy gradient flow are derived whose solutions yield the optimal LQR feedback gain for the augmented system and, consequently, the stabilizing LQI controller gain. The flow is shown to be a preconditioned version of the model-based policy gradient flow.
   \item For time-varying gains, as they arise when the flow is applied online, the direct, incremental, and feedforward realizations of the LQI controller are compared.
   \item The approach is validated on a distributed generation unit (DGU) in a DC microgrid, including an online implementation with closed-loop data under load changes, and by a Monte Carlo study on random systems with noisy data.
\end{enumerate}

\subsubsection*{Paper Organization}
Preliminaries are presented in Section~\ref{sec:prelim}. Section~\ref{math:main} introduces the data-driven parameterization, the convex program, and the policy gradient flow. Simulation results of a DGU are presented in Section~\ref{sec:sim}, and limitations are discussed in Section~\ref{sec:lim}.
Finally, the paper ends with a conclusion in Section~\ref{sec:conc}.

\subsubsection*{Notation}
\label{subsubsec:notation}
The set of real numbers is denoted by $\mathbb{R}$. The identity matrix of dimension $n \times n$ is given by $I_n$, and $0_{n,p}$ denotes the $n \times p$ zero matrix. For a matrix $A$, its transpose, trace, and Frobenius norm are denoted by $A^\top$, $\operatorname{tr}(A)$, and $\Vert A \Vert_F$, respectively. The operator $\operatorname{diag}(\cdot)$ constructs a block-diagonal matrix from its arguments. A symmetric positive definite (semidefinite) matrix $A$ is denoted by $A \succ 0$ ($A \succeq 0$). The Moore--Penrose pseudoinverse of a matrix $A$ is denoted by $A^\dagger$, and its nullspace is denoted by $\ker(A)$. The eigenvalues of a square matrix $A$ are denoted by $\lambda_i(A)$, and their real parts by $ \operatorname{Re} (\lambda_i)$. A square matrix is Hurwitz if all its eigenvalues have strictly negative real parts.

\section{Preliminaries}
\label{sec:prelim}
We consider the linear time-invariant system
\begin{subequations}
   \label{math:sys1}
   \begin{align}
 \dot x(t) &= Ax(t) + Bu(t), \label{math:sys1a} \\
 y(t) &= Cx(t),
\end{align}
\end{subequations}
where $x(t)\in\mathbb{R}^n$, $u(t)\in\mathbb{R}^m$, and $y(t)\in\mathbb{R}^p$.\footnote{Henceforth, explicit time dependency of variables is omitted for readability when clear from the context.}

\subsection{Linear Quadratic Regulator}
 The infinite-horizon LQR problem is formulated as%
 \begin{subequations}
   \label{math:LQR}
   \begin{align}
   \min_{x,u}~& \int_0^\infty x^\top Q x + u^\top R u \mathrm{d}t, \\
   \text{s.t.}~&  \eqref{math:sys1a} \text{ and }  x(0)=x_0,
\end{align}
 \end{subequations}
where $Q \succeq 0$ and $R \succ 0$.
\begin{assumption}
   \label{ass:1}
   The system $(A,B)$ is stabilizable and the pair $(A,\sqrt{Q})$ is detectable.
\end{assumption}

Assumption~\ref{ass:1} holds throughout.
The optimal  solution to \eqref{math:LQR} is the  controller $u = - K^* x = -R^{-1}B^\top P^* x$,  where $P^*$ is the unique positive semidefinite solution of the CARE
\begin{align}
 A^\top P + PA - PB R^{-1} B^\top P + Q = 0.
\end{align}
Moreover, the closed-loop matrix $A-BK^* $  is Hurwitz \cite[Secs.~3.1 and~3.2]{anderson2007}.
The optimal controller $K^*$ is also the solution to
\begin{subequations}
   \label{math:convex}
   \begin{align}
 \min_{K,W}~& \operatorname{tr }(QW ) + \operatorname{tr }(K^\top R K W ) \\
 \text{s.t.}~& (A-BK) W + W(A-BK)^\top + I_n \preceq 0,%
 & W\succeq 0,
\end{align}
\end{subequations}
which admits a convex reformulation  \cite{Feron}. %
The set of Hurwitz stable feedback gains of system $(A,B)$ is denoted by
\begin{align}
   \mathcal{K} & = \{ K \in \mathbb{R}^{m \times n} \mid \operatorname{Re}(\lambda_i(A-BK)) < 0 ~ \forall i \}, \label{math:K}
\end{align}
and is open, unbounded, and path-connected \cite[Section~3]{bu2019}.
The policy gradient flow of the LQR problem is defined as
\begin{subequations}
      \label{math:grad}
   \begin{align}
 \dot K & = - \nabla f_K, \quad K(0)=K_0\in\mathcal{K},\\
   \nabla f_K & = 2 (RK - B^\top P_K) W_K
\end{align}
\end{subequations}
where $P_K$ and $W_K$ are the solutions to the Lyapunov equations
\begin{align}
   0 & = (A-BK)^\top P + P (A-BK) + Q + K^\top R K, \\
   0 & = (A-BK) W + W (A-BK)^\top + I_n,
\end{align}
respectively. The gradient flow \eqref{math:grad} generates unique trajectories $K(t)$ within $\mathcal{K}$ that converge to $K^*$ for $t\to\infty$ \cite{bu2020clqr}. A related gradient flow that minimizes a continuous-time Bellman error instead of the LQR cost, and likewise generates stabilizing gains that converge to $K^*$, is proposed in \cite{giessler2025cdc}.

\subsection{Linear Quadratic Integral Control}
\label{subsec:LQI}
We present the LQI control from \cite{YOUNG01051972}. To enforce asymptotic tracking of a reference signal $r\in \mathbb{R}^p$, the integral state
\begin{align}
   \dot{z} = r-y = r-Cx,
\quad
z(t)=\int_0^t r(\tau)-y(\tau)\mathrm{d}\tau ,
\end{align}
is incorporated into the system \eqref{math:sys1}, yielding %
\begin{align}
   \label{math:sysa}
 \begin{bmatrix} \dot x \\ \dot z\end{bmatrix} & = \begin{bmatrix} A & 0_{n,p} \\ -C & 0_p  \end{bmatrix} \begin{bmatrix} x \\ z \end{bmatrix} + \begin{bmatrix} B \\ 0_{p,m} \end{bmatrix} u + \begin{bmatrix} 0_{n,p} \\ I_p \end{bmatrix} r.
\end{align}%
Define the aggregated state as $x_a := \begin{bmatrix} x \\ z \end{bmatrix}\in\mathbb{R}^{n+p}$.
A linear state-feedback control for \eqref{math:sysa} in \textit{direct form} $u = u_{\mathrm{dir}}$ is
\begin{align}
   \label{math:law_LQI}
 u_{\mathrm{dir}} = - K x_a = - K_{\mathrm{PD }} x - K_{\mathrm{I}} z,
\end{align}
where $K = \begin{bmatrix}K_{\mathrm{PD }}& K_{\mathrm{I}}\end{bmatrix}\in\mathbb{R}^{m\times (n+p)}$, and  $K_{\mathrm{PD}}\in\mathbb{R}^{m\times n}$ and $ K_{ \mathrm{I}}\in\mathbb{R}^{m\times p}$ denote the \textit{proportional and derivative} and \textit{integral} feedback gains, respectively.
For a constant reference $r$, let $\bar{x}_a = \begin{bmatrix}\bar{x}^\top & \bar{z}^\top\end{bmatrix}^\top$ be an equilibrium point with the corresponding equilibrium input $\bar{u}$, and define the error variables $\tilde{x}_a := x_a-\bar{x}_a$ and $\tilde{u} := u-\bar{u}$ (see \cite{MALKAPURE201455} for details). In error coordinates, the augmented dynamics \eqref{math:sysa} can be written as
\begin{align}
   \label{math:sys_a}
 \dot{\tilde x}_a & = \begin{bmatrix} A & 0_{n,p}  \\ - C & 0_p \end{bmatrix} \tilde x_a + \begin{bmatrix} B \\ 0\end{bmatrix} \tilde u = A_a \tilde x_a + B_a \tilde u.
\end{align}
Under $\tilde{u}=-K\tilde{x}_a$, the closed-loop system becomes
\begin{align}
   \label{math:closed_aug}
  \dot{\tilde x}_a & = \begin{bmatrix} A- B K_{\mathrm{PD}} & -B K_{\mathrm{I }} \\ - C & 0_p \end{bmatrix}  \tilde x_a  = (A_a - B_a K ) \tilde x_a.
\end{align}
If $A_a - B_a K$ is Hurwitz and the reference signal is constant, then $y(t)\to r$ for $t\to\infty$.
The LQI controller \eqref{math:law_LQI} for the original system \eqref{math:sys1} can be synthesized by computing the optimal LQR feedback $K^*$ for the augmented system \eqref{math:sys_a}.

Besides the direct form \eqref{math:law_LQI}, the gain $K$ can be realized in the \textit{incremental form} $u = u_{\mathrm{inc}}$ \cite[Section~3.4]{astrom1995pid} or in the \textit{feedforward form} $u = u_{\mathrm{ff}}$ \cite[Section~3.7.1]{kwakernaak_linear_1972}, where
\begin{align}
   \dot u_{\mathrm{inc}} &= -K \dot x_a, \quad u_{\mathrm{inc}}(0) = -K x_a(0), \label{math:law_inc}\\
   u_{\mathrm{ff}} &= \bar u - K_{\mathrm{PD}}(x-\bar x) - K_{\mathrm{I}} z, \label{math:law_ff}
\end{align}
and $(\bar x,\bar u)$ satisfies $A\bar x + B\bar u = 0$ and $C\bar x = r$.
All three forms yield the closed-loop matrix $A_a - B_aK$. For a constant gain $K$, \eqref{math:law_LQI} and \eqref{math:law_inc} coincide, whereas \eqref{math:law_ff} shifts the equilibrium of the integral state from $\bar z$ with $K_{\mathrm{I}}\bar z = -(\bar u + K_{\mathrm{PD}}\bar x)$ to $\bar z = 0$, so that a change of $r$ acts directly on the input.
For a time-varying gain $K(t)$ and a constant reference $r$, the three forms differ. Under \eqref{math:law_inc}, $w := \dot x_a$ obeys $\dot w = (A_a - B_aK(t))w$, and under \eqref{math:law_ff}, the error $\tilde x_a$ with respect to the equilibrium $\bar x_a = \left[\begin{smallmatrix} \bar x \\ 0 \end{smallmatrix}\right]$ obeys $\dot{\tilde x}_a = (A_a - B_aK(t))\tilde x_a$. Under \eqref{math:law_LQI}, in contrast, the equilibrium of the integral state depends on the gain, e.g., $\bar z(K) = -K_{\mathrm{I}}^{-1}(\bar u + K_{\mathrm{PD}}\bar x)$ for $m = p$, so that the error with respect to $\bar x_a = \left[\begin{smallmatrix} \bar x \\ \bar z(K(t)) \end{smallmatrix}\right]$ obeys $\dot{\tilde x}_a = (A_a - B_aK(t))\tilde x_a - \left[\begin{smallmatrix} 0_n \\ \dot{\bar z}(K(t)) \end{smallmatrix}\right]$, and a change of the gain excites a transient, which incremental PID algorithms avoid \cite[Section~3.7]{astrom1995pid}. %

\subsection{Data-driven Closed-loop Parameterizations}
We next summarize the closed-loop parameterization of \cite{de2019formulas} for continuous-time systems \cite[Remark~2]{de2019formulas}, adapted to the sample covariance parameterization in \cite{Zhao2025}.
Consider a sequence of $T\in\mathbb{N}$ measurements of state, input, and state derivative trajectories of system \eqref{math:sys1a}, sampled at interval $\Delta >0$,
\begin{subequations}
   \label{math:data}
   \begin{align}
 X &= \begin{bmatrix} x(0) & x(\Delta) & \dots & x((T-1)\Delta) \end{bmatrix}\in\mathbb{R}^{n\times T},\\
   U & =\begin{bmatrix} u(0) & u(\Delta) & \dots & u((T-1)\Delta)  \end{bmatrix}\in\mathbb{R}^{m\times T}, \\
   X' & = \begin{bmatrix} \dot x(0) & \dot x(\Delta) & \dots & \dot x((T-1)\Delta) \end{bmatrix}\in\mathbb{R}^{n\times T}.
\end{align}
\end{subequations}
These matrices satisfy
\begin{align}
   \label{math:sys_stacked}
 X' = AX + BU = \begin{bmatrix} B & A \end{bmatrix} \begin{bmatrix} U \\ X \end{bmatrix}.
\end{align}
Define the associated sample covariance matrices to \eqref{math:data} as%
\begin{align}
   \label{math:data_cov}
 \overline{X} & = \frac{1}{T }X \begin{bmatrix} U \\ X  \end{bmatrix}^\top\!\!,~ \overline{U}= \frac{1}{T}U \begin{bmatrix} U \\ X  \end{bmatrix}^\top \!\!,~ \overline{X}'= \frac{1}{T}X' \begin{bmatrix} U \\ X  \end{bmatrix}^\top \! \!,
\end{align}
whose dimensions are independent of the sample size $T$.

The following assumption ensures that the data fully characterizes the system and holds throughout the paper.
\begin{assumption}
   \label{ass:2}
   The matrix $\begin{bmatrix} \overline{U} \\ \overline{X}\end{bmatrix}$ has rank $n+m$.
\end{assumption}

For the integral data in Remark~\ref{remark:1}, whose sample covariance matrix has the same rank as $\left[\begin{smallmatrix} \tilde U \\ \tilde X\end{smallmatrix}\right]$, Assumption~\ref{ass:2} holds if $(A,B)$ is controllable and the input is constant on each interval with values that are persistently exciting of order $n+1$, and $\Delta$ avoids the isolated values at which sampling makes oscillatory modes of $A$ indistinguishable \cite[Definition~2 and Lemma~4]{Lopez2023}. This requires $T\geq (m+1)n+m$ intervals \cite[Remark~1]{Lopez2023}. Hence, if the order of the system is known, Assumption~\ref{ass:2} can be enforced a priori through the choice of the excitation input.

\begin{theorem}
   Let Assumption~\ref{ass:2} hold. Then, the closed-loop system $\dot x = (A-BK)x$ admits the equivalent data-driven representation
   \begin{align}
    \dot x &= \overline{X}' G x, \\
    \begin{bmatrix}  I_n \\- K  \end{bmatrix}& = \begin{bmatrix}  \overline{X} \\\overline{U}  \end{bmatrix} G,
   \end{align}
   where $G\in\mathbb{R}^{(n+m) \times n}$.
\end{theorem}

The proof follows from \cite[Theorem~2 and Remark~2]{de2019formulas} together with \cite[Section~III-A]{Zhao2025}, and is omitted for brevity.

\begin{remark}
   \label{remark:1}
Following \cite{Lopez2023}, to avoid state-derivative measurements and mitigate noise amplification caused by numerical differentiation, the data matrices \eqref{math:data} can be replaced by%
\begin{subequations}
   \label{math:int_data}
   \begin{align}
     \tilde{X} & \! =  \! \begin{bmatrix} \int_{t_1}^{t_1+\Delta }x(\tau)\mathrm{d}\tau & \dots &\int_{t_T}^{t_T+\Delta }x(\tau)\mathrm{d}\tau \end{bmatrix}, \\
      \tilde{U} &  \!= \! \begin{bmatrix} \int_{t_1}^{t_1+\Delta }u(\tau)\mathrm{d}\tau & \dots &\int_{t_T}^{t_T+\Delta }u(\tau)\mathrm{d}\tau \end{bmatrix},\\
          \overset{ \scriptscriptstyle \Delta}{X} &  \! =  \! \begin{bmatrix} x(t_1 \! + \!\Delta) \!- \!x(t_1) & \!  \!  \dots   \! \! \!  & x(t_T  \! +  \! \Delta) \!- \! x(t_T)\end{bmatrix},
\end{align}
\end{subequations}
respectively. These matrices result from integrating \eqref{math:sys1a} over the contiguous intervals $[t_i, t_i+\Delta]$ with $t_i=\Delta (i-1)$, $i\in\{1,\dots, T\}$, so that all measured data is used exactly once.
Hence, the linear relation \(\overset{\scriptscriptstyle \Delta}{X} = A \tilde{X} + B \tilde{U}\) holds,
and substituting  \eqref{math:data} by \eqref{math:int_data} to compute \eqref{math:data_cov} does not alter the subsequent results.
\end{remark}

\section{Problem Formulation and Approach}
\label{math:main}
In this section, we first analyze the augmented system \eqref{math:sys_a} and clarify the relation of LQI control to classical PI(D) control.
We then develop a data-driven closed-loop parameterization, followed by the formulation of the convex program and the policy gradient flow.
The main objective of the paper is the following.

\begin{objective}
Formulate a convex optimization problem and a policy gradient flow, using only data from the original system \eqref{math:sys1}, to synthesize the optimal LQR state-feedback gain for the augmented system \eqref{math:sys_a}.
\end{objective}

Let $Q_a = \operatorname{diag}(Q_x,Q_z)$ be the weighting matrix for the augmented state $\tilde{x}_a$, where $Q_x\succeq 0$ corresponds to original state $\tilde x$ and  $Q_z\succ 0$ to the integral state  $\tilde z$.
\footnote{Choosing $Q_z\succ0$ ensures that the zero eigenvalues introduced by the integrator dynamics are observable through $\sqrt{Q_z}$ and thus reflected in the cost, which is natural for the LQI problem.}

\begin{assumption}
   \label{ass:3}
   The pair $(A,B)$ is stabilizable, the pair $\bigl(A,\left[\begin{smallmatrix} C \\ \sqrt{Q_x}\end{smallmatrix}\right]\bigr)$ is detectable, and
   \begin{align}
    \operatorname{rank} \begin{bmatrix} A & B \\ C & 0_{p,m} \end{bmatrix} = n+p. \label{math:second_cond}
   \end{align}
\end{assumption}

By Lemmas~\ref{lem:stab} and~\ref{lem:det} below, Assumption~\ref{ass:3} holds if and only if $(A_a,B_a)$ is stabilizable and $(A_a,\sqrt{Q_a})$ is detectable, which ensures that the optimal LQR controller stabilizes the augmented system.

\begin{lemma}
   \label{lem:stab}
   The pair $(A_a,B_a)$ is stabilizable if and only if $(A,B)$ is stabilizable and \eqref{math:second_cond} holds.
\end{lemma}

\begin{proof}
   By the Popov--Belevitch--Hautus (PBH) test \cite[Theorem~14.3]{hespanha2018}, $(A_a,B_a)$ is stabilizable if and only if for every eigenvalue $\lambda\in\mathbb{C}$ of $A_a$ with $\operatorname{Re}(\lambda)\geq 0$ the matrix
   \begin{align}
    M(\lambda):= \! \begin{bmatrix}\lambda I_{n+p} - A_a  & \! B_a \end{bmatrix} = \begin{bmatrix} \lambda I_n -A & 0 & B \\ C & \! \lambda I_p & 0  \end{bmatrix}
   \end{align}
   has full row rank $n+p$.
   Consider the case $\lambda \neq 0$.
   Since $\lambda I_p$ is nonsingular, $\operatorname{rank} M(\lambda) = p + \operatorname{rank} \begin{bmatrix}  \lambda I_n - A & B  \end{bmatrix}$. Hence, $\operatorname{rank} M(\lambda) = n+p$ if and only if $\operatorname{rank} \begin{bmatrix}  \lambda I_n - A & B\end{bmatrix} =n$, which is precisely the condition for stabilizability of $(A,B)$.
 For the case $\lambda =0$, $M(0)$ has full row rank if and only if  $\operatorname{rank}M(0) = \operatorname{rank} \big[\begin{smallmatrix}A & B \\ C & 0 \end{smallmatrix}\big]=n+p$. Combining both cases completes the proof.
\end{proof}

For $m=p$, the corresponding classical result is stated in terms of the transfer matrix $C(sI_n-A)^{-1}B$ having no zeros at the origin \cite[Sec.~3.7.2]{kwakernaak_linear_1972}.
Condition \eqref{math:second_cond} implies $p\leq m$, $\operatorname{rank}(C)=p$, and $\operatorname{rank}(B)\geq p$. Thus, to satisfy Assumption~\ref{ass:3}, the number of control inputs must be at least as large as the number of outputs to be tracked. Moreover, \eqref{math:second_cond} ensures that $A\bar{x}+B\bar{u}=0$ and $C\bar{x}=r$ admit a solution $(\bar{x},\bar{u})$ for every constant reference $r$, which is unique for $m=p$. Hence, the error coordinates in Section~\ref{subsec:LQI} are well defined. %

\begin{lemma}
   \label{lem:det}
   The pair $(A_a,\sqrt{Q_a})$ is detectable if and only if $(A,\left[\begin{smallmatrix} C \\ \sqrt{Q_x}\end{smallmatrix}\right])$ is detectable.
\end{lemma}
\begin{proof}
   By the PBH test \cite[Theorem~16.6]{hespanha2018}, $(A_a,\sqrt{Q_a})$ is detectable if and only if for every eigenvalue $\lambda\in\mathbb{C}$ of $A_a$
    with $\operatorname{Re}(\lambda)\geq 0$ the matrix
   \begin{align}
     N = \begin{bmatrix} A-\lambda I_n & 0 \\ -C  &- \lambda I_p \\ \sqrt{Q_x} & 0 \\ 0 & \sqrt{Q_z }\end{bmatrix}
   \end{align}
   has full column rank.
   Let $\begin{bmatrix} x \\z \end{bmatrix}$ lie in the nullspace of $N$. Then, equivalently, $(A-\lambda I_n)x =0, -Cx - \lambda z =0, \sqrt{Q_x} x =0$ and $\sqrt{Q_z}z =0$. Since $Q_z\succ0$, $z=0$, hence $Cx=0$. Thus full column rank of $N$ holds if and only if
\begin{align}
\operatorname{rank}
\begin{bmatrix}
A-\lambda I_n\\
C\\
\sqrt{Q_x}
\end{bmatrix}
= n
\end{align}
for all $\operatorname{Re} (\lambda)\ge0$, i.e.,
$(A,\left[\begin{smallmatrix} C \\ \sqrt{Q_x}\end{smallmatrix}\right])$ is detectable.
\end{proof}

Note that Lemma~\ref{lem:det} states that no eigenvector of $A$ corresponding to an eigenvalue with nonnegative real part may lie simultaneously in the nullspaces of $C$ and $Q_x$.

\subsection{Relation to Proportional Integral Derivative Control}
\label{subsec:PID}
In this subsection, we clarify the relation of LQI control to PI(D) control. The PID control law is
\begin{align}
   \label{math:PID}
 u & = - K_{\mathrm{P}} (y-r) - K_{\mathrm{I}} \int_0^t r(\tau) - y(\tau)\mathrm{d}\tau - K_{\mathrm{D}} \frac{\mathrm{d} }{\mathrm{d} t}(y-r)
\end{align}
where $K_{\mathrm{P}}\in\mathbb{R}^{m\times p}$, $K_{\mathrm{I}}\in\mathbb{R}^{m\times p}$, and $K_{\mathrm{D}}\in\mathbb{R}^{m\times p}$, and falls into the class of dynamic output feedback controllers with feedthrough.
In error coordinates, the control law \eqref{math:PID} can be rewritten as
\begin{align}
 \tilde u = - K_{\mathrm{P}} C \tilde{x} - K_{\mathrm{I}} \tilde{z} - K_{\mathrm{D}} C \dot{\tilde{x}}. \label{math:PID_law}
\end{align}
Substituting the system dynamics \eqref{math:sys1a} into \eqref{math:PID_law} yields
\begin{align}
 (I_m + K_{\mathrm{D}}C B) \tilde u = - (K_{\mathrm{P}}C + K_{\mathrm{D}}CA )\tilde{ x} - K_{\mathrm{I}} \tilde{z}.
\end{align}
Hence, the PID control law can be represented by a state-feedback controller of the augmented system \eqref{math:sys_a}
\begin{align}
 \tilde{u} = - \tilde{K}_{\mathrm{PD}} \tilde{x} - \tilde{K}_{\mathrm{I}} \tilde{z} = -\tilde{K} \tilde x_a,
\end{align}
where $\tilde{K}=\begin{bmatrix}\tilde{K}_{\mathrm{PD}} & \tilde{K}_{\mathrm{I}}\end{bmatrix}\in\mathbb{R}^{m\times (n+p)}$, and
\begin{subequations}
   \label{math:nonlinear}
   \begin{align}
 \tilde{K}_{\mathrm{PD}}&= (I_m + K_{\mathrm{D}} C B)^{-1} (K_{\mathrm{P}} C + K_{\mathrm{D}} C A ),\\
     \tilde{K}_{\mathrm{I}}&= (I_m + K_{\mathrm{D}} C B)^{-1} K_{\mathrm{I}}.
\end{align}
\end{subequations}
Note that the derivative action in the original control law \eqref{math:PID_law} is absorbed into the proportional feedback gain $\tilde K_{\mathrm{PD}}$ of the augmented system \eqref{math:sys_a}.
This also explains the presence of the derivative component in $K_{\mathrm{PD}}$ in \eqref{math:law_LQI}.
Following \cite{mukhopadhyay1978pid}, stabilizing PID gains can be obtained by first  solving the LQR problem for the augmented system to obtain $\tilde{K}$ and subsequently reconstructing the original PID gains by solving the nonlinear equations \eqref{math:nonlinear}.
However, the nonlinear mapping \eqref{math:nonlinear} introduces significant difficulties, as there may exist a finite number, infinitely many, or no PID gains satisfying the equations. For single-input single-output systems, the original PID gains can always be reconstructed if the input and output matrices are orthogonal \cite{polyak2022new}.

We now consider the PI case, obtained by setting $K_{\mathrm{D}} = 0$.
Analogous to \eqref{math:closed_aug}, the augmented closed-loop system under \eqref{math:PID} with $K_{\mathrm{D}} =0$ in error coordinates is
\begin{align}
  \dot{\tilde x}_a & = A_a \tilde{x}_a - B_a \hat{K} \tilde{x}_a  \\
  & = \begin{bmatrix} A - B K_{\mathrm{P}} C & - B K_{\mathrm{I}} \\ -C & 0 \end{bmatrix}\tilde{x}_a, \label{math:closed_matrix}
\end{align}
where $\hat{K}=\begin{bmatrix}  K_{\mathrm{P}} C & K_{\mathrm{I}}\end{bmatrix} \in\mathbb{R}^{m \times (n+p)}$. Hence, PI control can be interpreted as a structurally constrained state-feedback controller for the augmented system \eqref{math:sys_a}.
The gain block acting on the original state $\tilde{x}$ must factor as $ K_{\mathrm{P}}C $, which restricts it to the row space of $C$ and reflects the general limitation of static output feedback compared with state feedback. In contrast, the optimal  LQI gain $K^* = \begin{bmatrix} K_{\mathrm{PD}}^* & K_{\mathrm{I}}^* \end{bmatrix}$ obtained from the LQR problem for $(A_a,B_a)$ generally does not admit such a factorization.\footnote{We note that if $C$ is square and invertible ($p=n$), the constraint $K_{\mathrm{PD}} = K_{\mathrm{P}}C$ imposes no restriction, and PI control recovers full-state feedback LQI control.} Unless $K_{\mathrm{PD}}^*$ happens to lie in the row space of $C$, the achievable performance under PI control is strictly inferior to that of the full-state-feedback LQI design.

This structural restriction has consequences for controller synthesis. The set of stabilizing PI gains may fail to be path-connected \cite[Section~4.3]{datta1999structure}, so that a policy gradient method cannot reach gains outside the connected component of its initial gain.
Moreover, no convex reparameterization of the problem \eqref{math:convex} is known for the closed-loop matrix in \eqref{math:closed_matrix}, which reflects the general difficulty of synthesizing fixed-structure output feedback controllers such as PI(D) controllers \cite{datta1999structure}.

\subsection{Data-driven Parameterization}
Next, we introduce a data-driven closed-loop parameterization of the augmented system \eqref{math:closed_aug}. Analogous to \eqref{math:data} and \eqref{math:data_cov}, we define the data matrices
\begin{align}
   Y & = \begin{bmatrix}  y(0) &  y(\Delta) & \dots &  y((T-1)\Delta) \end{bmatrix}\in\mathbb{R}^{p\times T},\\
   \overline{Y} & = \frac{1}{T}Y \begin{bmatrix} U \\ X  \end{bmatrix}^\top \in\mathbb{R}^{p\times (n+m)},
\end{align}
which satisfy $\overline{Y} = C \overline{X}$.
For the integral variant mentioned in Remark~\ref{remark:1}, $Y$ has to be substituted by
\begin{align}
   \label{math:data_Y}
 \tilde{Y} &=  \begin{bmatrix} \int_{t_1}^{t_1+\Delta }y(\tau)\mathrm{d}\tau & \dots &\int_{t_T}^{t_T+\Delta }y(\tau)\mathrm{d}\tau \end{bmatrix}\in\mathbb{R}^{p\times T}.
\end{align}

\begin{theorem}
   \label{th:data}
   The augmented closed-loop system \eqref{math:closed_aug} admits the equivalent data-driven representation
   \begin{subequations}
      \label{math:data_cl}
      \begin{align}
    \dot{\tilde{x}}_a &= \begin{bmatrix} \overline{X}' \\ -\overline{Y}\end{bmatrix} G \tilde{x}_a, \\
    \begin{bmatrix}I_n & 0_{n,p} \\ -K_{\text{PD}} & - K_{\text{I}}\end{bmatrix} & = \begin{bmatrix} \, \overline X \,\, \\ \overline{U}\end{bmatrix} G, \label{math:cl_2}
   \end{align}
   \end{subequations}
   where $G\in\mathbb{R}^{(n+m)\times (n+p)}$. Particularly, $K = - \overline{U}G$.
\end{theorem}
\begin{proof}
   The closed-loop matrix of \eqref{math:closed_aug} is
\begin{align}
   A_{\mathrm{cl}}  = \begin{bmatrix} A & B \\ -C & 0_{p,m} \end{bmatrix} \begin{bmatrix} I_n & 0_{n,p} \\ - K_{PD} & -K_I\end{bmatrix}.
\end{align}
Moreover, due to \eqref{math:sys_stacked} and $\overline{Y} \! \!= \! C \overline{X}$, the measured data~satisfies
\begin{align}
 \begin{bmatrix}\overline{X}' \\ -\overline{Y}\end{bmatrix} = \begin{bmatrix} A & B \\ -C & 0_{p,m} \end{bmatrix} \begin{bmatrix}\, \overline X \,\, \\ \overline{U}\end{bmatrix}. \label{math:data_1}
\end{align}
By Assumption~\ref{ass:2}, the square matrix $\left[\begin{smallmatrix} \overline{X} \\ \overline{U}\end{smallmatrix}\right]$ is invertible. Therefore, for any $K = \begin{bmatrix}K_{\mathrm{PD }}& K_{\mathrm{I}}\end{bmatrix}  \in \mathbb{R}^{m \times (n+p)}$, there exists a unique $G$ satisfying \eqref{math:cl_2}.
Hence,%
\begin{align}
 A_{\mathrm{cl}} \overset{ \eqref{math:cl_2} }{=} \begin{bmatrix} A & B \\ -C & 0_{p,m} \end{bmatrix}   \begin{bmatrix} \overline X \\ \overline{U}\end{bmatrix} G \overset{\eqref{math:data_1}}{=}  \begin{bmatrix} \overline{X}' \\ -\overline{Y}\end{bmatrix} G.
\end{align}
Finally, %
$\overline{U} G = \begin{bmatrix} - K_{\mathrm{PD}} & - K_{\mathrm{I}}\end{bmatrix} = -K.$
\end{proof}

Hence, every closed-loop matrix $A_a-B_aK$ admits a data-driven parameterization through some $G$ satisfying $\begin{bmatrix}I_n & 0_{n,p}\end{bmatrix} = \overline{X} G$.

\begin{remark}
   Note that the augmented closed-loop \eqref{math:closed_aug} can be represented by data matrices $\overline{X},\overline{U},\overline{X}'$, and $\overline{Y}$ of the underlying system \eqref{math:sys1} without measuring the integral state $\tilde{z}$.
   Hence, the controller $K$ can be designed based on data, but each of the control laws \eqref{math:law_LQI}--\eqref{math:law_ff} in original coordinates requires the integrator $z(t) = \int_0^t r(\tau)-y(\tau) \mathrm{d}\tau$.
   Moreover, the equilibrium $(\bar x, \bar u)$ in \eqref{math:law_ff} can be obtained from the same data. By \eqref{math:data_1}, any $g\in\mathbb{R}^{n+m}$ with
   \begin{align}
      \begin{bmatrix} \overline{X}' \\ \overline{Y}\end{bmatrix} g = \begin{bmatrix} 0_{n,1} \\ r \end{bmatrix} \label{math:eq_data}
   \end{align}
   yields $\left[\begin{smallmatrix} \bar x \\ \bar u \end{smallmatrix}\right] = \left[\begin{smallmatrix} \overline{X} \\ \overline{U} \end{smallmatrix}\right] g$ with $A\bar x + B\bar u = \overline{X}'g = 0$ and $C\bar x = \overline{Y} g = r$. For $m = p$, the matrix in \eqref{math:eq_data} is invertible by \eqref{math:second_cond} and Assumption~\ref{ass:2}, so that $g$ is unique.
\end{remark}

\subsection{Convex Program}
Next, we adapt the convex program \eqref{math:convex} to the data-driven representation \eqref{math:data_cl} of the augmented system \eqref{math:sys_a}.

\begin{theorem}
   \label{th:conv}
   The optimal LQR feedback gain of system \eqref{math:sys_a} is  $K^* = -\overline{U} Z^*(W^*)^{-1}$, where $Z^*\in\mathbb{R}^{(n+m)\times (n+p)}$ and $W^*\in\mathbb{R}^{(n+p)\times (n+p)}$ minimize the convex program
   \begin{subequations}
      \label{math:convex_data}
       \begin{align}
      \min_{W,Z,S}~& \operatorname{tr}\left(Q_a W  \right) + \operatorname{tr}\left(S  \right)\\
      \text{s.t.}~ & \begin{bmatrix} S & R^{1/2} \overline{U} Z \\ Z^\top \overline{U}^\top R^{1/2} & W \end{bmatrix}\succeq 0, \label{math:c11}  \\
    & \begin{bmatrix} \overline{X}' \\ -\overline{Y }\end{bmatrix} Z + Z^\top \begin{bmatrix} \overline{X}' \\ -\overline{Y }\end{bmatrix}^\top + I_{n+p} \preceq 0,\\
    & \begin{bmatrix} I_{n} & 0_{n,p}\end{bmatrix} W = \overline{X} Z, \\
    & W\succeq 0,
   \end{align}
   with $S = S^\top \in\mathbb{R}^{m\times m}$.
   \end{subequations}
\end{theorem}
\begin{proof}
   Substituting \eqref{math:data_cl} into \eqref{math:convex} and using  the cyclic property of the trace operator yields%
   \begin{subequations}
      \begin{align}
 \min_{W,G}~& \operatorname{tr }(Q_a W ) + \operatorname{tr }(R^{1/2} \overline{U}G W G^\top \overline{U}^\top  R^{1/2}) \\
 \text{s.t.}~& \begin{bmatrix}\overline{X}' \\ -\overline{Y}\end{bmatrix} G W + W \Bigl(  \begin{bmatrix} \overline{X}' \\ -\overline{Y }\end{bmatrix} G\Bigr)^\top + I_{n+p} \preceq 0, \label{math:c1}  \\
 & \begin{bmatrix}I_n & 0_{n,p}\end{bmatrix} = \overline{X} G, \\
 & W \succeq 0.
\end{align}
   \end{subequations}
By substituting $Z  = G W $, the constraint \eqref{math:c1} becomes affine in $(Z,W)$. However, the cost term $ \operatorname{tr}\!\big(R^{1/2} \overline{U} Z W^{-1} Z^\top \overline{U}^\top R^{1/2}\big)$ remains nonconvex. To obtain the convex program \eqref{math:convex_data}, we introduce the epigraph variable $S = S^\top$ and upper bound this term as
\begin{align*}
\operatorname{tr}\!\big(R^{1/2} \overline{U} Z W^{-1} Z^\top \overline{U}^\top R^{1/2}\big) \le \operatorname{tr}(S),
\end{align*}
which is enforced via \eqref{math:c11},
obtained by a Schur complement argument. Every feasible $W$ is positive definite. If $Wv=0$, then \eqref{math:c11} implies $\overline{U}Zv = 0$, so that, with $\begin{bmatrix} I_n & 0_{n,p}\end{bmatrix}W = \overline{X}Z$ and \eqref{math:data_1}, multiplying the second constraint from both sides by $v$ yields $\Vert v\Vert^2 \leq 0$, i.e., $v=0$. Hence, $G^*  =  Z^*(W^*)^{-1}$ follows from the substitution, and the optimal gain is recovered as $K^* = -\overline{U} Z^*(W^*)^{-1}$.
\end{proof}

Note that every feasible solution of the convex program \eqref{math:convex_data} yields a feedback gain such that $A_a-B_a K$ is Hurwitz.

\subsection{Projected Policy Gradient Flow}
In this subsection, we formulate the projected policy gradient flow for the data-driven representation \eqref{math:data_cl}. In contrast to the convex program \eqref{math:convex_data}, which yields $K^*$ in one shot, the gain $K(t)$ along the gradient flow evolves continuously and remains stabilizing for all $t\geq 0$ (see Theorem~\ref{th:proj}). This makes the gradient flow suited for an online implementation that adapts the gain when the data, and thus the closed-loop parameterization, are updated during operation, as in the discrete-time setting of \cite{Zhao2025}.
Let $\mathcal{G}$ denote the set of all $G\in\mathbb{R}^{(n+m)\times(n+p)}$ satisfying \eqref{math:cl_2} and rendering the closed-loop matrix Hurwitz, i.e.,
\begin{align}
   \label{math:G_set}
   \mathcal{G} &:=
\left\{
G \in \mathbb{R}^{(n+m)\times(n+p)} \;\middle|\;
\begin{aligned}
& \overline{X}G =
\begin{bmatrix}
I_n & 0_{n,p}
\end{bmatrix}, \\
& \begin{bmatrix}
\overline{X}' \\
-\overline{Y}
\end{bmatrix} G \text{ is Hurwitz}
\end{aligned}
\right\} \\
& \; = \left\{G \in \mathbb{R}^{(n+m)\times(n+p)} \mid \exists K\in\mathcal{K}_a \text{ s.t. } \eqref{math:cl_2}\right\},
\end{align}
where $\mathcal{K}_a$ denotes the set of gains $K\in\mathbb{R}^{m\times(n+p)}$ for which $A_a-B_aK$ is Hurwitz, defined analogously to \eqref{math:K}.

Analogous to the model-based case in \cite{bu2020clqr}, the LQR cost function of the system \eqref{math:data_cl} parameterized in terms of $G$ is defined as the matrix function
\begin{align}
   \label{math:lqr_cost}
 f_G:\mathcal{G}\to\mathbb{R},~~ G\mapsto \operatorname{tr}(P_G),
\end{align}
where $P_G\in\mathbb{R}^{(n+p)\times(n+p)}$ is the solution of the Lyapunov equation
\begin{align}
   \label{math:lyap_P}
 0 & = \!\Bigl(\begin{bmatrix} \overline{X}' \\ -\overline{Y }\end{bmatrix}G\Bigr)^{\!\!\top} P_G + P_G \begin{bmatrix} \overline{X}' \\- \overline{Y }\end{bmatrix}G  + Q_a + (\overline{U}G)^\top R (\overline{U}G).
\end{align}
By Theorem~\ref{th:data} and Assumption~\ref{ass:2}, the map $G\mapsto K = -\overline{U}G$ is an affine bijection from $\mathcal{G}$ onto $\mathcal{K}_a$, and $f_G$ is the LQR cost of the augmented system \eqref{math:sys_a} for the gain $K$. Hence, $f_G$ inherits the properties of the model-based LQR cost shown in \cite[Proposition~3.2 and Lemma~3.3]{bu2020clqr}. In particular, $f_G$ is real analytic on $\mathcal{G}$ and coercive, i.e., $f_G \to \infty$ as $G\to\partial \mathcal{G}$ or $\Vert G\Vert_F\to\infty$, so that its sublevel sets are compact.

\begin{proposition}
   Let $G\in\mathcal{G}$. Then, the gradient of the LQR cost \eqref{math:lqr_cost} with respect to $G$ is
   \begin{align}
    \nabla f_G = 2 (\overline{U}^\top R \overline{U}G + \begin{bmatrix}\overline{X}'\\ -\overline{Y} \end{bmatrix}^\top P_G)W_G,
    \label{math:grad_}
   \end{align}
   where $P_G$ and $W_G = W_G^\top \in\mathbb{R}^{(n+p)\times(n+p)}$ satisfy \eqref{math:lyap_P} and
   \begin{align}
      \label{math:lyap_W}
    0 & = \begin{bmatrix}\overline{X}'\\ -\overline{Y} \end{bmatrix} G W_G + W_G \bigl(\begin{bmatrix}\overline{X}'\\ -\overline{Y} \end{bmatrix} G\bigr)^\top + I_{n+p},
   \end{align}
   respectively.
\end{proposition}

\begin{proof}
   The function \eqref{math:lqr_cost} is a composition of smooth functions, i.e., $G\mapsto P_G \mapsto \operatorname{tr}(P_G)$. The differential of \eqref{math:lyap_P} with respect to $G$ is given by
   \begin{align}
    0 & = \Bigl(\begin{bmatrix} \overline{X}' \\ -\overline{Y }\end{bmatrix}G\Bigr)^{\!\!\top} dP_G +  dP_G \Bigl( \begin{bmatrix} \overline{X}' \\- \overline{Y }\end{bmatrix}G \Bigr)  + F^\top + F, \label{math:LyapB}
   \end{align}
   where $F= dG^\top (\overline{U}^\top R \overline{U}G + \left[\begin{smallmatrix}\overline{X}'\\ -\overline{Y} \end{smallmatrix}\right]^\top P_G)$.\footnote{For matrix differentials, see \cite[Section~1]{minka1997old}.}
   Since $G\in\mathcal{G}$, the solution to \eqref{math:LyapB} is
   \begin{align}
    dP_G = \int_0^\infty   e^{\big(\left[\begin{smallmatrix}\overline{X}'\\ -\overline{Y} \end{smallmatrix}\right] G\big)^{\top} t}  ( F + F^\top  )
    e^{ \left[\begin{smallmatrix}\overline{X}'\\ -\overline{Y} \end{smallmatrix}\right] G t }
    \mathrm{d}t.
   \end{align}
   Thus, by using the cyclic invariance of the trace operator, the differential of the LQR cost is
   \begin{align}
    df_G &= \operatorname{ tr}(dP_G ) = 2  \operatorname {tr} \big( F \int_0^\infty e^{ \left[\begin{smallmatrix}\overline{X}'\\ -\overline{Y} \end{smallmatrix}\right] G t }  e^{\big(\left[\begin{smallmatrix}\overline{X}'\\ -\overline{Y} \end{smallmatrix}\right] G\big)^{\top} t} \mathrm{d}t  \big). \label{math:int}
   \end{align}
   The integral in \eqref{math:int} is the solution to \eqref{math:lyap_W}, yielding $df_G = 2 \operatorname{tr}(F W_G )$, implying \eqref{math:grad_}.
\end{proof}

To ensure that the linear equality constraint $\overline{X}G =
\begin{bmatrix}
I_n & 0_{n,p}
\end{bmatrix}$ is preserved along the gradient flow, we project the gradient $\nabla f_G$ onto the tangent space of $\mathcal{G}$, i.e., onto $\operatorname{ker}(\overline{X})$, using the orthogonal projection
\begin{align}
 \Pi :=  I_{n+m} - \overline{X}^\dagger \overline{X},
\end{align}
where $\overline{X}^\dagger = \overline{X}^\top (\overline{X} \overline{X}^\top)^{-1}$ is the right inverse of $\overline{X}$. As an orthogonal projection, $\Pi$ is symmetric and idempotent, implying its eigenvalues lie in $\{0,1\}$ and $\Pi\succeq 0$ \cite[1.1.P5]{horn2012matrix}.

The following lemma relates the projected gradient to the gradient of the LQR cost with respect to $K$ and establishes uniqueness of the stationary point of $f_G$ on $\mathcal{G}$.

\begin{lemma}
   \label{lem:unique}
   Let $G\in\mathcal{G}$ and $K = -\overline{U}G$. Then
   \begin{align}
    \Pi\nabla f_G = -\Pi\overline{U}^\top \nabla f_K, \label{math:grad_K}
   \end{align}
   where $\nabla f_K = 2(RK - B_a^\top P_G)W_G$ is the gradient of the LQR cost of the augmented system \eqref{math:sys_a} with respect to $K$. Moreover, $\Pi\nabla f_G = 0$ holds if and only if $G = G^*$, where $K^* = -\overline{U}G^*$ is the optimal LQR gain of \eqref{math:sys_a}.
\end{lemma}

\begin{proof}
   For $G\in\mathcal{G}$, Theorem~\ref{th:data} yields $\left[\begin{smallmatrix}\overline{X}' \\ -\overline{Y}\end{smallmatrix}\right] G = A_a - B_aK$ and $\overline{U}G = -K$, so that \eqref{math:lyap_P} and \eqref{math:lyap_W} are the Lyapunov equations of the augmented closed loop, where $W_G\succ 0$, since $A_a-B_aK$ is Hurwitz. By \eqref{math:data_1},
   \begin{align*}
    \begin{bmatrix}\overline{X}' \\ -\overline{Y}\end{bmatrix}^{\!\top} P_G = \overline{X}^\top \begin{bmatrix} A^\top & -C^\top\end{bmatrix} P_G + \overline{U}^\top B_a^\top P_G,
   \end{align*}
   so that \eqref{math:grad_} becomes $\nabla f_G = 2\overline{X}^\top \begin{bmatrix} A^\top & -C^\top\end{bmatrix} P_G W_G - \overline{U}^\top\nabla f_K$. Since $\Pi\overline{X}^\top = 0$, \eqref{math:grad_K} follows.
   If $\Pi\overline{U}^\top v = 0$ for some $v\in\mathbb{R}^m$, then $\overline{U}^\top v = \overline{X}^\top w$ for some $w\in\mathbb{R}^n$, which implies $v = 0$, since $\left[\begin{smallmatrix}\overline{X} \\ \overline{U}\end{smallmatrix}\right]$ is invertible by Assumption~\ref{ass:2}. Hence, $\Pi\nabla f_G = 0$ if and only if $\nabla f_K = 0$, i.e., $RK = B_a^\top P_G$, since $W_G\succ 0$. Substituting $RK = B_a^\top P_G$ into \eqref{math:lyap_P} yields the CARE of \eqref{math:sys_a}, of which $P_G$ is the stabilizing solution, since $A_a-B_aK$ is Hurwitz. This solution is unique under Assumption~\ref{ass:3}, so that $K = R^{-1}B_a^\top P_G = K^*$. Finally, $G^*$ is unique, since $\left[\begin{smallmatrix}\overline{X} \\ \overline{U}\end{smallmatrix}\right] G = \left[\begin{smallmatrix} I_n & 0_{n,p} \\ -K_{\mathrm{PD}} & -K_{\mathrm{I}} \end{smallmatrix}\right]$ determines $G$ uniquely for every $G\in\mathcal{G}$.
\end{proof}

\begin{theorem}
   \label{th:proj}
   The projected gradient flow
   \begin{align}
      \label{math:proj_grad}
    \dot G = -\alpha \Pi \nabla f_G, \quad G(0)\in\mathcal{G}, \quad \alpha>0
   \end{align}
   converges to $G^*$, i.e., $G(t)\to G^*$ for $t\to\infty$.
\end{theorem}
\begin{proof}
   Let $V(G) = f_G - f_{G^*}$ be the Lyapunov candidate.
   Since $f_G$ is the LQR cost of \eqref{math:sys_a} for the gain $K=-\overline{U}G$, whose unique minimizer is $K^*$, and $G\mapsto K$ is injective on $\mathcal{G}$ (see the proof of Lemma~\ref{lem:unique}), $V(G^*)=0$ and $V(G)>0$ for all $G\in\mathcal G\setminus\{G^*\}$.
    The time derivative of $V(G)$ along \eqref{math:proj_grad}~is
   \begin{align}
    \dot V(G) &= \operatorname{tr} ( \nabla f_G^\top \dot G  ) = - \alpha \operatorname{tr}( \nabla f_G^\top \Pi \nabla f_G )  \\
    & = - \alpha \Vert \Pi^{1/2} \nabla f_G \Vert_F^2 \leq 0.
   \end{align}
   By Lemma~\ref{lem:unique}, $\dot V (G)=0$ holds only for $G=G^*$.  Since $\dot V(G)<0 $ for all $G\in\mathcal{G}\setminus \{G^*\}$, $V(G)$ is a Lyapunov function, implying asymptotic stability of $G^*$. Moreover, $G(t)$ remains in the compact sublevel set $\{ G\in\mathcal{G}\mid f_G \leq f_{G(0)}\}$ for any $G(0)\in\mathcal{G}$, because $V(G)$ is nonincreasing along $G(t)$. %
   It follows that $G(t)$ yields stabilizing feedback gains $K(t)$ for all $t \geq 0$.
Consequently, the region of attraction of $G^*$ is $\mathcal{G}$.
\end{proof}

\begin{remark}
   By Lemma~\ref{lem:unique}, the gain $K(t) = -\overline{U}G(t)$ along \eqref{math:proj_grad} obeys $\dot K = -\overline{U}\dot G = -\alpha H\nabla f_K$ with $H := \overline{U}\Pi\overline{U}^\top\succ 0$, where $H\succ 0$ holds since $\Pi\overline{U}^\top$ has full column rank. Hence, \eqref{math:proj_grad} is a preconditioned version of the model-based policy gradient flow \eqref{math:grad} for the augmented system \eqref{math:sys_a} with the data-dependent preconditioner $H$.
\end{remark}

\section{Simulation Results}
\label{sec:sim}
We consider a single bus of a DC microgrid that comprises a DGU with a controller and a constant-impedance load, following \cite{Nahata.2020}.
The control architecture resembles an LQI controller, where the states are fed back proportionally, and an additional integral state is introduced to ensure asymptotic tracking of the reference voltage.
The closed-loop system is depicted in Fig.~\ref{fig:draw}.
The objective is to regulate the bus voltage $v$ to a reference voltage $r$ via the input voltage $u$ of the buck converter. The load is modeled as a constant admittance $Y$, while the buck converter is represented by an averaged model consisting of a controllable voltage source followed by an RLC filter. Filter and load parameters are
$R_{\mathrm{f}}=\SI{0.2}{\ohm}$, $L_{\mathrm{f}}=\SI{2}{\milli\henry}$, $C_{\mathrm{f}}= \SI{2}{\milli\farad}$, and $Y = \SI{0.02}{\siemens}$.

 \begin{figure}[thb]
   \centering
   \scalebox{0.732}{
   \begin{tikzpicture}[scale=0.89]
      \draw (0,0) to [V=$u$] (0,-2);
      \draw (0,0) to [R,l=$R_{\mathrm{f}}$] (2,0)
      to [L,l=$L_{\mathrm{f}}$] (4,0) to [short,i=$i$](4.5,0) to [short,-*](4.8,0)
      to [C,l_=$C_{\mathrm{f}}$] (4.8,-2) to [short,*-] (0,-2);
      \draw (4.8,0) to [short,-o] (7.3,0);
      \draw (4.8, -2) to [short,*-o] (7.3,-2);
      \draw (7.3,0) to [open, v^=$~v$,voltage=straight](7.3,-2);
      \draw (6.5,0) to [R, l_=$Y$, *-*] (6.5,-2);
         \node[draw, label=above:{DGU}](Bat) at (3.2,-0.7) [rectangle,dashed,minimum width=6.95cm,minimum height=3cm] {};
      \node[isosceles triangle,draw,isosceles triangle apex angle=60,shape border rotate=180, inner sep = 2pt](Gain) at (2.3,-3.5){$K$};
      \draw [-latex] (Gain.apex) -| (0.95,-1.25);
      \node[draw,circle,inner sep=0pt, minimum size = 3pt] (I) at (3.9,0){};
      \draw [-] (I) -- (3.5,-1);
      \draw [-latex] (3.5,-1) |- ([yshift=+0.4cm]Gain.lower side);
      \draw [-] (8.1,-1) --(8.5,-1) |- (5,-3.1) -- (4.25,-3.5);
      \node[draw,circle,fill=black,minimum size = 2pt,inner sep=0pt] (circ) at (8.5,-3.1){};
      \draw [-latex] (4.25,-3.5) --(Gain.lower side);
      \node[draw,circle,inner sep=0pt, minimum size = 13pt] (sum) at (8.5,-3.9){};
      \draw[-latex] (8.5,-3.1) -- (sum);
      \node[draw, rectangle, minimum width = 13pt, minimum height = 13pt] (int) at (8.4,-2,5){};
      \draw[-] (int.south east) -- (int.north west);
      \draw[-latex] (int.west) -- node[pos=0.3,above,align=center] {$z$}([yshift=-0.4cm]Gain.lower side);
      \draw[-latex] (sum.west) -- (int.east);
      \node (Vref) at (10,-3.9){$r$};
      \draw[-latex] (Vref) -- (sum);
      \node[align=left] at (8.63,-3.9){\footnotesize +};
      \node[align=left] at (8.5,-3.79){\small -};
      \node[draw, label=below:{Grid-forming controller}](control) at (5.65,-3.55) [rectangle,dashed,minimum width=9cm,minimum height=1.65cm] {};
   \end{tikzpicture}
   }
   \caption{DGU in closed loop with the LQI controller.}
   \label{fig:draw}
\end{figure}

The DGU with the load obeys the open-loop dynamics
\begin{align}
   \label{math:DGU}
 \begin{bmatrix} \dot v \\ \dot i \end{bmatrix}& = \begin{bmatrix} - \frac{1}{C_{\mathrm{f}}} Y & \frac{1}{C_{\mathrm{f}}} \\ -\frac{1}{L_{\mathrm{f}}} & - \frac{1}{L_{\mathrm{f}}} R_{\mathrm{f}}\end{bmatrix} \begin{bmatrix}v \\ i \end{bmatrix} + \begin{bmatrix} 0 \\ \frac{1}{L_{\mathrm{f}}}\end{bmatrix}u,
\end{align}
where $v\in\mathbb{R}$ denotes the bus voltage and $i\in\mathbb{R}$ the filter current. The measured output is the bus voltage, yielding the output matrix $C =  \begin{bmatrix} 1 & 0 \end{bmatrix}$.
The DGU controller has the LQI structure $u= -K_{\mathrm{PD}}  \left[\begin{smallmatrix} v \\ i \end{smallmatrix} \right]- K_{\mathrm{I}}\int_0^t r(\tau)-v(\tau) \mathrm{d}\tau$ in the direct form \eqref{math:law_LQI}, where $K_{\mathrm{PD}} \in\mathbb{R}^{1\times 2}$ and $K_{\mathrm{I}}\in\mathbb{R}$. For the LQR design of the augmented system, the weighting matrices are chosen as $R = 1$ and $Q_a = \operatorname{diag}(1,1,100)$ to achieve fast tracking of the voltage reference.

 Since the system parameters are assumed to be unknown, the required data matrices are obtained from input-state-output measurements, during which the load is constant. Time-varying loads are addressed by the online implementation in Fig.~\ref{fig:3}. To collect the data, starting from $x(0) = 0$, a random input signal is applied to the system, which is piecewise constant over intervals of $\SI{0.02}{\second}$ with independent, normally distributed values of zero mean and standard deviation $\SI{100}{\volt}$. A total of $T=10$ samples are recorded with $\Delta = \SI{0.1}{\second}$, i.e., over contiguous intervals covering the first second. From these measurements, the data matrices  \eqref{math:int_data} and \eqref{math:data_Y} are constructed to compute the sample covariance matrices $\overline{X},\overline{U}, \overline{X}'$, and $\overline{Y}$. By Theorem~\ref{th:data}, these matrices can be used to parameterize the closed-loop system. Since the input is piecewise constant, the states and their integrals are computed by an exact discretization of \eqref{math:DGU} with the matrix exponential, so that the data is free of integration errors. As a ground truth, the optimal gain $K^\star$ is computed from the exact augmented model using the \texttt{lqr} function in MATLAB.

 Using the proposed data-driven approaches of Theorem~\ref{th:conv} and Theorem~\ref{th:proj}, we compute the feedback gain $K=\begin{bmatrix}K_{\mathrm{PD }}& K_{\mathrm{I}}\end{bmatrix}$ that minimizes the LQR objective of the augmented system \eqref{math:sys_a}.
In Fig.~\ref{fig:1}, the voltage $v(t)$ and the input $u(t)$ are shown. Within the first second, data is collected in open loop as described above. After the first second, the problem \eqref{math:convex_data} is solved to obtain $K^* \approx \begin{bmatrix}0.409 & 1.163 & -10.000\end{bmatrix}$, where $\Vert K^\star - K^* \Vert_F = 3.6 \times 10^{-6}$ is due to the solver tolerance.
From $\SI{1}{\second}$ to $\SI{4}{\second}$, the system is controlled in closed loop with the gain $K^*$, where the reference voltage $r(t)$ changes from $\SI{400}{\volt}$ to $\SI{600}{\volt}$ to $\SI{200}{\volt}$ at the time instances $t=\SI{2}{\second}$ and $t=\SI{3}{\second}$. The trajectory $v(t)$ in Fig.~\ref{fig:1} shows that the reference voltage is tracked by the bus voltage.

\begin{figure}
   \centering
   \includegraphics[scale=0.9]{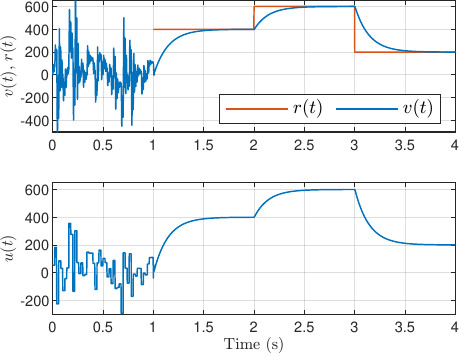}
   \caption{Trajectories of the bus voltage $v(t)$ and input $u(t)$ during open-loop data collection and closed-loop control.}
   \label{fig:1}
\end{figure}

\begin{figure}
   \centering
   \includegraphics[scale=0.85]{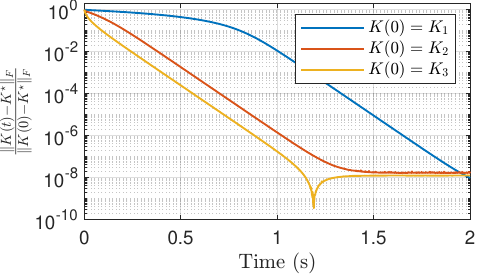}
   \caption{Normalized residuals   $\frac{\Vert K(t) - K^\star \Vert_F }{\Vert K(0) - K^\star \Vert_F}$ of the projected gradient flow for initial gains $K(0)\in\{K_1,K_2,K_3\}$.}
   \label{fig:2}
\end{figure}

In Fig.~\ref{fig:2}, the normalized residuals $\frac{\Vert K(t) - K^\star \Vert_F }{\Vert K(0) - K^\star \Vert_F}$ of the projected policy gradient flow \eqref{math:proj_grad} are shown for the initial stabilizing gains $K_1 = \begin{bmatrix}0.5 & 0.1 & -50\end{bmatrix} $, $K_2 = \begin{bmatrix} 5 & 1 & -15\end{bmatrix}$, and $K_3 = \begin{bmatrix} 0 & 0 & -1\end{bmatrix}$.
The initial values $G(0)$ are computed using \eqref{math:cl_2}, while the trajectories $G(t)$ are obtained from the projected gradient flow \eqref{math:proj_grad}, yielding the depicted trajectory $K(t) = - \overline{U} G(t)$. The learning rate $\alpha = 10^{9}$ only scales the time axis.
For all initializations, the residuals converge linearly, for $K_1$ after an initial phase, consistent with the model-based setting \cite{bu2020clqr}, until rounding errors in the evaluation of the data-driven gradient dominate at about $10^{-8}$.

\begin{figure}
   \centering
   \includegraphics[scale=0.85]{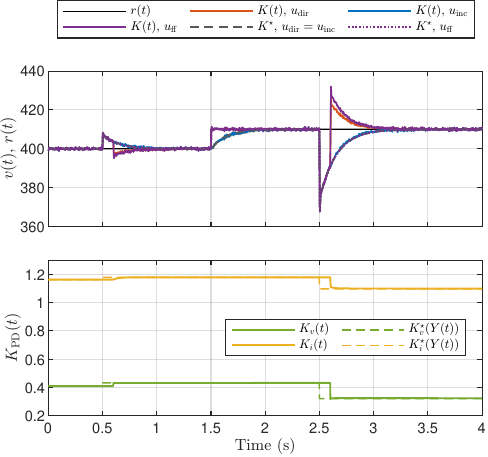}
   \caption{Online data-driven policy gradient flow under load and reference changes. The top panel shows the voltage $v(t)$ for the adaptive gain $K(t)$ and the fixed nominal gain $K^\star$ in the direct, incremental, and feedforward forms, and the bottom panel shows $K_{\mathrm{PD}}(t)$ in the incremental form and the optimal gain $K^\star_{\mathrm{PD}}(Y(t))$ for the current load.}
   \label{fig:3}
\end{figure}

Fig.~\ref{fig:3} shows the projected policy gradient flow \eqref{math:proj_grad} applied online with closed-loop data under time-varying loads. The load admittance changes from $Y = \SI{0.02}{\siemens}$ to $Y = \SI{0.001}{\siemens}$ at $t=\SI{0.5}{\second}$ and to $Y = \SI{0.1}{\siemens}$ at $t=\SI{2.5}{\second}$, and the reference voltage changes from $\SI{400}{\volt}$ to $\SI{410}{\volt}$ at $t=\SI{1.5}{\second}$. Starting from the nominal gain $K^\star$ at $t=0$, with the data window filled by closed-loop operation under $K^\star$ and the exploration signal before $t=0$, the gain $K(t)$ follows \eqref{math:proj_grad} with $\alpha = 2\times 10^{14}$, where the data matrices are computed from the last $20$ contiguous intervals of length $\Delta = \SI{5}{\milli\second}$, and the input is perturbed by a piecewise constant exploration signal with standard deviation $\SI{1}{\volt}$ and hold time $\SI{1}{\milli\second}$. Both times are of the order of the closed-loop time constants of about $\SI{3}{\milli\second}$, since for longer intervals the state follows the exploration signal quasi-statically and the matrix in Assumption~\ref{ass:2} becomes nearly singular. Since \eqref{math:data_1} is linear in the data, the data is centered before computing the sample covariances, which removes the operating point. Between two data updates, the data matrices are constant, and $K(t) = -\overline{U}G(t)$ is obtained by integrating \eqref{math:proj_grad}. While the data window contains measurements from two different loads, detected by $\operatorname{rank}\bigl[\tilde X^\top~\tilde U^\top~\overset{\scriptscriptstyle \Delta}{X}{}^\top~\tilde Y^\top\bigr]^\top > n+m$, the flow is paused, which occurs for $\SI{0.1}{\second}$ after each load change. The equilibrium in \eqref{math:law_ff} is computed from the same data via \eqref{math:eq_data}, whereas the fixed gain $K^\star$ in the feedforward form uses the nominal equilibrium.

After each pause, $K_{\mathrm{PD}}(t)$ reaches the optimal gain $K^\star_{\mathrm{PD}}(Y(t))$ of the current load up to $10^{-2}$ within $\SI{25}{\milli\second}$, while $K_{\mathrm{I}}(t)$ stays within $2.2\times 10^{-2}$ of the optimal integral gain $K^\star_{\mathrm{I}} = -10$, which is independent of the load. Since $K^\star_{\mathrm{PD}}$ changes by less than $0.12$ in each entry between the loads, the adaptive gain stays close to $K^\star$. The three realizations of Section~\ref{subsec:LQI} differ only when $K(t)$ or the equilibrium changes. During the pause after $t=\SI{2.5}{\second}$, all controllers yield the same voltage drop to about $\SI{368}{\volt}$. In the incremental form, the voltage deviates by at most $\SI{1.3}{\volt}$ from that of the fixed gain $K^\star$ due to the exploration signal. In the direct form, the gain update at $t=\SI{2.6}{\second}$ shifts the equilibrium of the integral state and causes an overshoot to $\SI{426.6}{\volt}$. In the feedforward form, the reference change acts directly on the input, so that $v(t)$ reaches $\SI{409.9}{\volt}$ at $t=\SI{1.6}{\second}$ compared to $\SI{404.9}{\volt}$ in the direct form, but the update of $(\bar x, \bar u)$ at $t=\SI{2.6}{\second}$ causes an overshoot to $\SI{432.0}{\volt}$, since the integral state has already compensated the load change.
The resulting controller is a nonlinear dynamic state feedback, and no formal stability guarantee is provided for the time-varying data and loads. Between two load changes and for constant data, the incremental and feedforward forms yield the linear time-varying error dynamics of Section~\ref{subsec:LQI} with the gain $K(t)$ of the policy gradient flow, as considered in the stability analysis of the LQR policy gradient flow in closed loop with linear time-varying systems in \cite{gießler2025dynamicstatefeedbackcontrollpv}.

\begin{figure}
   \centering
   \includegraphics[scale=0.85]{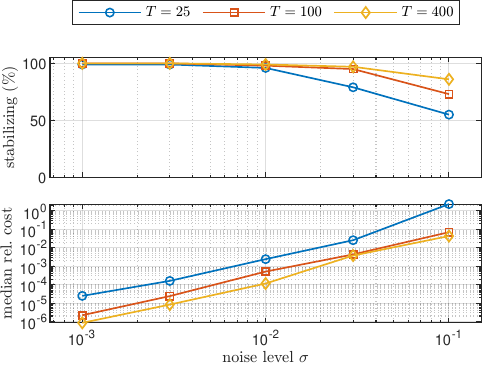}
   \caption{Share of stabilizing gains and median relative cost increase $(f(K)-f(K^\star))/f(K^\star)$ of the solution of \eqref{math:convex_data} versus the noise level $\sigma$ for $T$ intervals in the Monte Carlo study on $100$ random systems with $n = 6$ and $m = p = 2$.}
   \label{fig:mc}
\end{figure}

To assess higher-dimensional systems and noisy data, the convex program \eqref{math:convex_data} is applied to $100$ random open-loop stable systems with $n = 6$ and $m = p = 2$, so that the data can be collected in open loop. The entries of $A$, $B$, and $C$ are normally distributed, $A$ is shifted such that its spectral abscissa is uniformly distributed in $[-1,-0.1]$, and the smallest singular value of the matrix in \eqref{math:second_cond} is at least $0.05$. The weights are $Q_a = I_{n+p}$ and $R = I_m$. For each system, data is collected over $T\in\{25,100,400\}$ contiguous intervals of length $\Delta = \SI{1}{\second}$ with a random input that is piecewise constant over $\SI{0.2}{\second}$ and scaled such that the noise-free state has unit root mean square value. The data is corrupted by process noise of intensity $\sigma^2$ \cite[Definition~1.30]{kwakernaak_linear_1972} and by measurement noise of variance $\sigma^2$ on the state and output, sampled every $\SI{10}{\milli\second}$, so that $\sigma$ is the noise-to-signal ratio of the measurements. Fig.~\ref{fig:mc} shows the share of stabilizing gains and the median relative cost increase. For noise-free data, all gains are optimal up to a relative cost increase of $2\times 10^{-6}$. With noise, the cost increase grows with $\sigma$ and decreases with $T$. For $\sigma = 10^{-2}$, its median is $2.4\times 10^{-3}$, $5.1\times 10^{-4}$, and $1.2\times10^{-4}$ for $T = 25$, $100$, and $400$, respectively. For $\sigma = 10^{-1}$, $55\,\%$, $73\,\%$, and $86\,\%$ of the gains are stabilizing, and increasing $T$ from $100$ to $400$ reduces the median cost increase only by a factor of $1.6$, which reflects the bias discussed in Section~\ref{sec:lim}. Up to $\sigma = 10^{-2}$, all non-stabilizing gains occur in runs where the solver MOSEK reports numerical problems. In all other runs, the gain agrees with the certainty-equivalence gain up to a relative deviation of $1.6\times 10^{-2}$.

\section{Limitations}
\label{sec:lim}
The results rely on noise-free data. Since $\left[\begin{smallmatrix} \overline{X} \\ \overline{U}\end{smallmatrix}\right]$ is invertible, every $G$ satisfying \eqref{math:cl_2} yields the closed-loop matrix $\left[\begin{smallmatrix} \overline{X}' \\ -\overline{Y}\end{smallmatrix}\right] G = \hat M \left[\begin{smallmatrix} I_n & 0_{n,p} \\ -K_{\mathrm{PD}} & -K_{\mathrm{I}} \end{smallmatrix}\right]$, where $\hat M := \left[\begin{smallmatrix} \overline{X}' \\ -\overline{Y}\end{smallmatrix}\right]\left[\begin{smallmatrix} \overline{X} \\ \overline{U}\end{smallmatrix}\right]^{-1}$ is the least-squares estimate of the matrix in \eqref{math:data_1}. Hence, also for noisy data, the solution of \eqref{math:convex_data} and the stationary point of \eqref{math:proj_grad} coincide with the certainty-equivalence gain, which is the LQR gain for the augmented model defined by $\hat M$, analogous to \cite[Lemma~1]{Zhao2025}. Noise therefore affects the gain only through $\hat M$. For the integral data of Remark~\ref{remark:1}, measurement noise enters $\tilde{X}$ only through its integral over each interval, which averages it \cite[Section~1.6.3]{garnier2008direct}, whereas it enters $\overset{\scriptscriptstyle \Delta}{X}$ as the difference of the noise at the interval endpoints, which is not averaged. Process noise enters $\overset{\scriptscriptstyle \Delta}{X}$, while $\tilde{X}$ contains the response to this noise on the same interval, so that the noise is correlated with the regressor and biases the estimate \cite[Section~2.6]{larsson2008estimation}. This bias does not vanish as $T$ grows \cite[Section~7.3]{ljung1999}, which motivates instrumental variables \cite[Section~7.6]{ljung1999} or the regularizations in \cite{dorfler2022,ZhaoCovariance}.

For open-loop unstable systems, the data cannot be collected in open loop over a long horizon. Since \eqref{math:data_1} holds for any input, the data can instead be collected in closed loop with a stabilizing gain and an exploration signal, as in Fig.~\ref{fig:3}, provided that Assumption~\ref{ass:2} holds. In Fig.~\ref{fig:3}, this required an exploration signal and an interval length of the order of the closed-loop time constants. A stabilizing initial gain is required by the policy gradient flow anyway.

Finally, the matrix $\left[\begin{smallmatrix} \overline{X} \\ \overline{U}\end{smallmatrix}\right]$ squares the condition number of $\left[\begin{smallmatrix} \tilde{U} \\ \tilde{X}\end{smallmatrix}\right]$. This limits the accuracy of the data-driven gradient, see Fig.~\ref{fig:2}, and renders the convex program \eqref{math:convex_data} numerically sensitive for higher-order systems with poorly excited directions.

\section{Conclusion}
\label{sec:conc}
This paper introduced a data-driven approach for the synthesis of LQI controllers for continuous-time systems. Using a closed-loop data-driven parameterization, we derived a convex optimization problem that enables the computation of the optimal LQR feedback gain of the augmented system directly from measured data. In addition, a projected policy gradient flow, which is a preconditioned version of the model-based policy gradient flow, was introduced to compute the optimal controller within the set of stabilizing gains. For time-varying gains, the incremental and feedforward realizations of the LQI controller yield dynamics with the closed-loop matrix $A_a - B_aK(t)$, in contrast to the direct realization, and the equilibrium for the feedforward realization is obtained from the same data. The approach was illustrated using a DGU with an LQI controller in a DC microgrid, including an online implementation with closed-loop data under load changes, and by a Monte Carlo study on random systems with noisy data. Future work will focus on consistent estimators and robust formulations for noisy data, as well as on a stability analysis of the online implementation with time-varying data.

\printbibliography

\end{document}